\documentclass[11pt]{article}
\usepackage{graphicx}
\usepackage{algorithm}
\usepackage{algpseudocode}
\usepackage{amsfonts}

\providecommand{\U}[1]{\protect\rule{.1in}{.1in}}

\usepackage{tikz}
\usetikzlibrary{shapes}
\usepackage{pgfplots}
\usepackage{verbatim}
\usepgfplotslibrary{statistics}
\pgfplotsset{compat=1.8}
\usetikzlibrary{positioning}
\usepackage{amsmath,amssymb,amsthm}
\usepackage{subcaption}
\usepackage[export]{adjustbox}

\usepackage[normalem]{ulem}

\renewcommand\and{\end{tabular}\kern-\tabcolsep\ and\ \kern-\tabcolsep\begin{tabular}[t]{c}}
\let\origthanks\thanks
\renewcommand\thanks[1]{\begingroup\let\rlap\relax\origthanks{#1}\endgroup}

\usepackage{array}

\newtheorem{theorem}{Theorem}

\newtheorem{example}[theorem]{Example}

\newtheorem{assumption}{Assumption}

\newtheorem{proposition}[theorem]{Proposition}

\newcommand{\quotes}[1]{``#1''}

\begin{document}

\title{Predicting Dynamics on Networks Hardly Depends on the Topology}
\author{Bastian Prasse\thanks{Faculty of Electrical Engineering, Mathematics and
Computer Science, P.O Box 5031, 2600 GA Delft, The Netherlands; \emph{email}:
b.prasse@tudelft.nl, p.f.a.vanmieghem@tudelft.nl} \and Piet Van Mieghem\footnotemark[1]}
\date{Delft University of Technology\\
May 29, 2020}
\maketitle
\begin{abstract}
Processes on networks consist of two interdependent parts: the network topology, consisting of the links between nodes, and the dynamics, specified by some governing equations. This work considers the prediction of the future dynamics on an unknown network, based on past observations of the dynamics. For a general class of governing equations, we propose a prediction algorithm which infers the network as an intermediate step. Inferring the network is impossible in practice, due to a dramatically ill-conditioned linear system. Surprisingly, a highly accurate prediction of the dynamics is possible nonetheless: Even though the inferred network has no topological similarity with the true network, both networks result in practically the same future dynamics.
\end{abstract}

\section{Introduction}
The interplay of dynamics and structure lies at the heart of myriad processes on networks, ranging from predator-prey interactions on ecological networks \cite{may2001stability} and epidemic outbreaks on physical contact networks \cite{pastor2015epidemic} to brain activity on neural networks \cite{cabral2014exploring}. To relate the network structure and the process dynamics, there are two approaches of opposing directions. On the one hand, a great body of research \cite{boccaletti2006complex, barrat2008dynamical, porter2016dynamical} focusses on the question: \textit{What is the impact of the network structure on the dynamics of a process?} For instance, the impact of the network of online social media friendships on the spread of fake news. On the other hand, network reconstruction methods \cite{timme2014revealing, wang2016data, newman2018network, peixoto2019network} consider the inverse problem: \textit{Given some observations of dynamics, what can we infer about the network structure?} As an example, one may ask to determine the path of an infectious virus from one individual to another, given observations of the epidemic outbreak. 

The prediction of dynamics on an \textit{unknown} network seems to require the combination of both directions: first, the reconstruction of the network structure based on past observations of the dynamics and, second, the estimation of the future dynamics based on the inferred (i.e. reconstructed) network. Intuitively, one may expect that an accurate prediction of the dynamics is possible only if the network reconstruction is accurate. In this work, paradoxically, we show the contrary: it is possible to accurately predict a general class of dynamics without the network structure. 

\section{Modelling Dynamics on Networks} \label{sec:dynmical_models}
The network is represented by the $N \times N$ weighted adjacency matrix~$A$ whose elements are denoted by~$a_{ij}$. If there is a directed link from node~$j$ to node~$i$, then it holds that $a_{ij}>0$, and $a_{ij}=0$ otherwise. Throughout this work, we make a clear distinction between the \textit{network topology} and the \textit{interaction strengths} \cite{barrat2004architecture}. The network topology, or network structure, is the set of all links: all node pairs $(i,j)$ for which $a_{ij}>0$. If there is a link from node $j$ to node $i$, then the interaction strength is specified by the link weight $a_{ij}$. For instance, consider the two $3\times 3$ adjacency matrices 
\begin{align*}
A = \begin{pmatrix}
0 & 1 & 0 \\
2 & 0 & 0 \\
0 & 2 & 2
\end{pmatrix}, \quad \hat{A} = \begin{pmatrix}
0 & 2 & 0 \\
1 & 0 & 0 \\
0 & 2 & 3
\end{pmatrix}.
\end{align*} 
For all nodes $i,j$, it holds that $a_{ij}>0$ if and only if $\hat{a}_{ij}>0$. Hence, the two matrices $A$ and $\hat{A}$ have the same network topology. However, the interaction strength, e.g., from node 2 to node 1 is different, because $a_{12}=1$ but $\hat{a}_{12}=2$.

We denote the \textit{nodal state} of node~$i$ at time~$t$ by $x_i(t)$ and the nodal state vector by $x(t) = (x_1(t), ..., x_N(t))^T$. We consider a general class of dynamical models on networks \cite{barzel2013universality, timme2014revealing, laurence2019spectral} that describe the evolution of the nodal state~$x_i(t)$ of any node $i$ as
\begin{align} \label{eq:interaction_model}
\frac{d x_i(t)}{d t} = f_i\left( x_i(t) \right) + \sum^N_{j=1} a_{ij} g\left(x_i(t), x_j(t)\right).
\end{align}
The function~$f_i\left( x_i(t) \right)$ describes the \textit{self-dynamics} of node~$i$. The sum in (\ref{eq:interaction_model}) represents the interactions of node~$i$ with its neighbours. The interaction between two nodes $i$ and $j$ depends on the adjacency matrix~$A$ and the \textit{interaction function}~$g\left(x_i(t), x_j(t)\right)$. A broad spectrum of models follows from (\ref{eq:interaction_model}) by specifying the self-dynamics function~$f_i$ and the interaction function~$g$. We study six particular models of dynamics on networks, which are summarised by Table~\ref{table:dynamical_models}: 
\begin{center}
  \begin{table*} 
  \centering
 \setlength{\extrarowheight}{3pt}
  \begin{tabular}{ | l | c | c | }
    \hline   
Model & $f_i\left( x_i(t) \right)$ & $g\left(x_i(t), x_j(t)\right)$ \\ \hline\hline
Lotka-Volterra (LV) & $x_i(t)( \alpha_i - \theta_i x_i(t))$ & $- x_i(t) x_j(t)$ \\ \hline 
Mutualistic population (MP) & $x_i(t)( \alpha_i - \theta_i x_i(t))$ &$x_i(t) x^2_j(t) (1 + x^2_j(t))^{-1}$  \\ \hline 
Michaelis-Menten (MM) & $- x_i(t)$ & $ x^h_j(t)(1+x^h_j(t))^{-1}$ \\ \hline 
SIS epidemics (SIS) & $-\delta_i x_i(t)$ & $( 1- x_i(t)) x_j(t)$ \\ \hline 
Kuramoto (KUR) & $\omega_i$ & $\sin\left( x_i(t) - x_j(t) \right)$ \\ \hline  
Cowan-Wilson (CW) & $-x_i(t)$ & $\left( 1+\exp\left(-\tau ( x_j(t)- \mu )\right) \right)^{-1}$ \\ \hline  
 \end{tabular}
 \caption{Models of dynamics on networks.\label{table:dynamical_models}} 
  \end{table*}
\end{center}
\begin{description}
\item[Lotka-Volterra population dynamics (LV)] The Lotka-Volterra model \cite{macarthur1970species} describes the population dynamics of competing species. The nodal state $x_i(t)$ denotes the population size of species~$i$, the growth parameters of species~$i$ equal $\alpha_i>0$ and $\theta_i >0$, and the link weight~$a_{ij}$ quantifies the competition rate, or predation rate, of species~$j$ on species~$i$. 

 \item[Mutualistic population dynamics (MP)] We adopt the model of Harush and Barzel \cite{harush2017dynamic} to describe mutualistic population dynamics. The nodal state~$x_i(t)$ denotes the population size of species~$i$, the  growth parameters of species~$i$ are denoted by $\alpha_i>0$ and $\theta_i >0$, and the link weight~$a_{ij}>0$ quantifies the strength of mutualism between species~$i$ and species~$j$.

\item[Michaelis-Menten regulatory dynamics (MM)] The dynamics of gene regulatory networks can be described by the Michaelis-Menten equation \cite{alon2006introduction, gao2016universal, harush2017dynamic}. The nodal state~$x_i(t)$ is the expression level of gene~$i$, the Hill coefficient is denoted by $h$, and the link weights~$a_{ij}>0$ are the reaction rate constants. 

\item[Susceptible-Infected-Susceptible epidemics (SIS)] Spreading phenomena, such as the epidemic of an infectious disease, can be described by the susceptible-infected-susceptible model \cite{bailey1975mathematical, lajmanovich1976deterministic, van2009virus, pastor2015epidemic}. The nodal state~$x_i(t)$ equals the infection probability of node~$i$. The parameter~$\delta_i >0$ denotes the curing rate, and the link weight~$a_{ij}$ is the infection rate from node~$j$ to node~$i$.

\item[Kuramoto oscillators (KUR)] The Kuramoto model \cite{kuramoto2003chemical} has been applied to various synchronisation phenomena of phase oscillators, such as fMRI activity of brain regions \cite{cabral2014exploring}. Here, the nodal state~$x_i(t)$ corresponds to the phase of oscillator~$i$, the parameter $\omega_i$ denotes the natural frequency of node~$i$, and the coupling strength from node~$j$ to node~$i$ is given by the link weight~$a_{ij}$.

\item[Cowan-Wilson neural firing (CW)] The firing-rates of neurons can be described by the Cowan-Wilson model \cite{wilson1972excitatory, laurence2019spectral}. Here, the nodal state~$x_i(t)$ is the activity of neuron~$i$, and the parameters $\tau$ and $\mu$ are the slope and the threshold of the neural activation function. The link weight~$a_{ij}$ specifies the number and strength of synapses from neuron $j$ to neuron $i$.
\end{description}

As stated in \cite{strogatz2001exploring}, there are three possibilities for the qualitative long-term behaviour of the dynamical system (\ref{eq:interaction_model}). First, the nodal state $x(t)$ might approach a steady state $x_\infty = \underset{t\rightarrow \infty}{\operatorname{lim}} ~x(t)$. At the steady state $x_\infty$, the nodal state does not change any longer, and it holds that $dx(t)/dt=0$. Second, the nodal state $x(t)$ might converge to a \textit{limit cycle}, a curve on which the nodal state $x(t)$ circulates forever. Third, the nodal state $x(t)$ might never come to rest nor enter a repeating cycle. Then, the state $x(t)$ perpetually continues to move in an irregular pattern.

\section{Prediction Algorithm for Dynamics on Networks} \label{sec:prediction_algorithm}
The true adjacency matrix $A$ is unknown. To predict the nodal state $x(t)$, we obtain an estimate $\hat{A}$ of the matrix $A$ from past observations of the nodal state $x(t)$. With the estimated matrix $\hat{A}$, we can approximate the governing equations (\ref{eq:interaction_model}) for the nodal state $x(t)$. Figure~\ref{fig:framework} illustrates the framework for predicting network dynamics.

\begin{figure}[!ht]
\centering
  \begin{subfigure}[c]{0.3\textwidth}
         \centering
         \caption{Observation}
         \tikz[remember picture]{        
        \node(1){
         \includegraphics[width=\textwidth, frame]{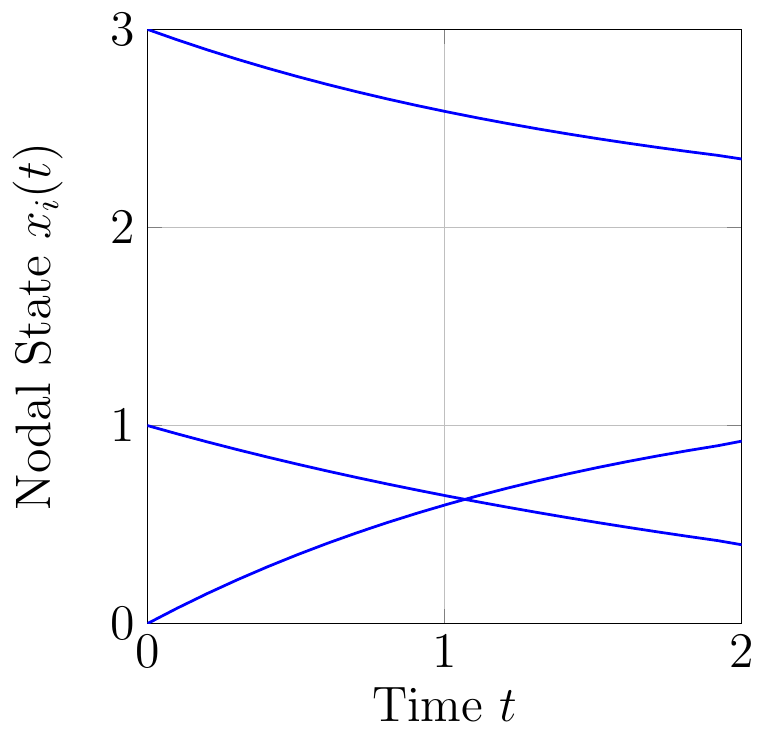}
         };  }     
     \end{subfigure}
     \hfill
     \begin{subfigure}[c]{0.3\textwidth} 
     \centering
     \caption{Reconstructed network}
     \tikz[remember picture]{      
     \node(2){  
     \fbox{
        $\hat{A} = \begin{pmatrix}
        \hat{a}_{11} & \hat{a}_{12} & \hat{a}_{13} \\
        \hat{a}_{21} & \hat{a}_{22} & \hat{a}_{23} \\
        \hat{a}_{31} & \hat{a}_{3 2} & \hat{a}_{33} 
              \end{pmatrix}$ 
              }
     };  } 
    \end{subfigure}
     \hfill
     \begin{subfigure}[c]{0.3\textwidth}
         \centering
         \caption{Prediction}\tikz[remember picture]{        
        \node(3){
         \includegraphics[width=\textwidth, frame]{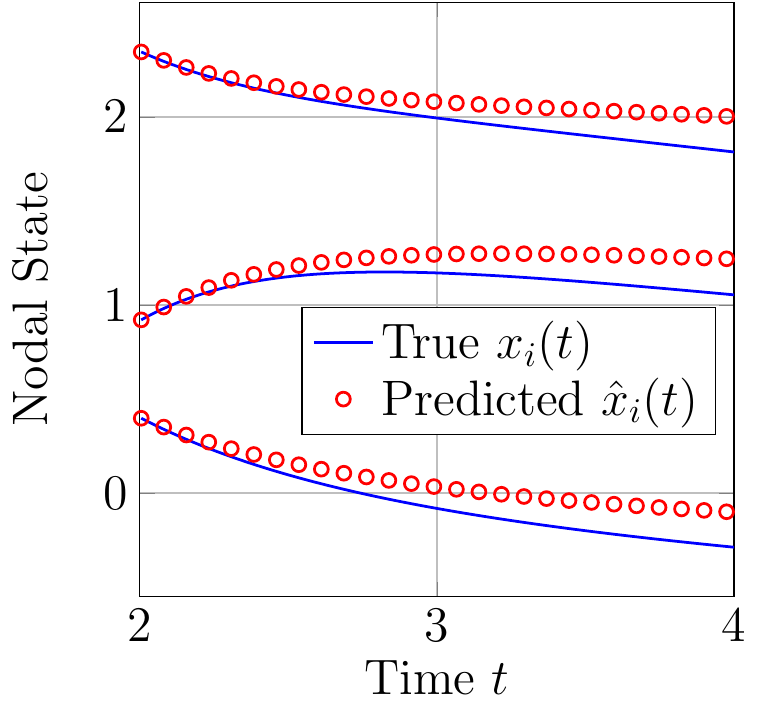}
          };  }     
     \end{subfigure}
\tikz[overlay,remember picture]{
\draw[line width=2pt, ->] ([xshift=0cm]1.east) -- ([xshift=0.75cm]1.east)node[midway,left,text width=0.35\textwidth]{};
\draw[line width=2pt, <-] ([xshift=-0.25cm]3.west) -- ([xshift=-1cm]3.west)node[midway,left,text width=0.35\textwidth]{};
} 
        \caption{\textbf{Framework for Predicting Dynamics on Networks.} This example shows a small network of $N=3$ nodes. \mbox{\textbf{(a)} The} nodal state $x_i(t)$ is observed for all nodes $i$ until the observation time $t_\textrm{obs}=2$. The evolution of the state $x_i(t)$ follows from the system~(\ref{eq:interaction_model}) with the known functions $f_i$, $g$ and the unknown adjacency matrix $A$. \mbox{\textbf{(b)} From} the nodal state observations, we infer an estimate~$\hat{A}$ for the true adjacency matrix $A$ by the LASSO (\ref{eq:lasso}). \mbox{\textbf{(c)} For} any time $t\ge t_\textrm{obs}$, the predicted nodal state $\hat{x}_i(t)$ follows from the system~(\ref{eq:interaction_model}) by replacing the true adjacency matrix $A$ with the estimate $\hat{A}$. The predicted nodal state is initialised as $\hat{x}_i(t_\textrm{obs}) = x_i(t_\textrm{obs})$ for all nodes $i$. }
        \label{fig:framework}
\end{figure}
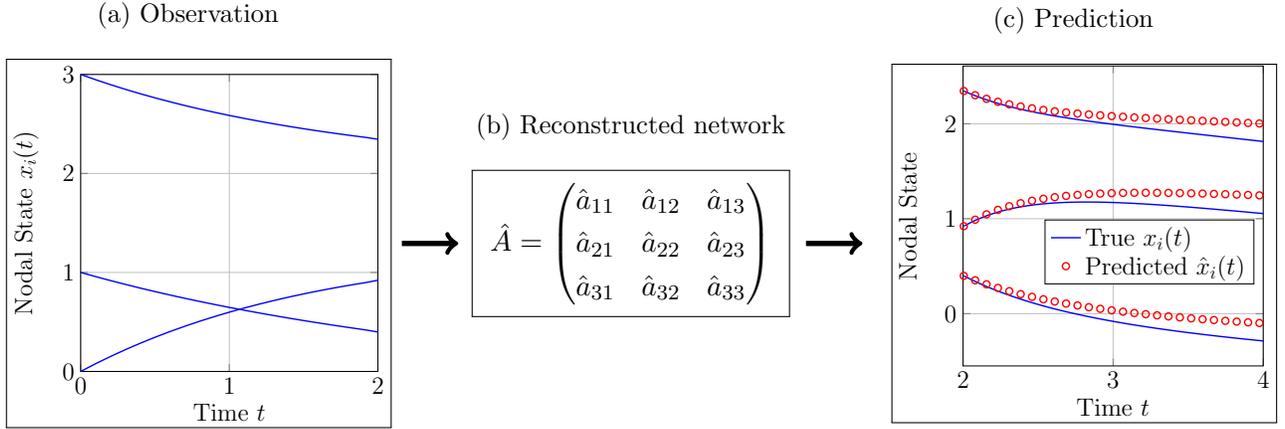

 We consider $n+1$ nodal state observations $x(0), x(\Delta t), ..., x(n \Delta t)$ from the initial time $t=0$ until the observation time $t = t_\textrm{obs}$. Here, $\Delta t>0$ denotes the sampling time with $n \Delta t = t_\textrm{obs}$. For a sufficiently small sampling time $\Delta t$, the solution of the model (\ref{eq:interaction_model}) obeys
\begin{equation}\label{eq:derivative_apx}
x_i\left( (k+1) \Delta t\right) \approx x_i\left( k \Delta t\right) + \Delta t \left.\frac{dx_i(t)}{dt}\right|_{t = k\Delta t} 
\end{equation}
at every time $k=0, ..., n-1$. A crucial observation is that the derivative $dx_i(t)/dt$ in (\ref{eq:interaction_model}) is linear with respect to the entries $a_{ij}$ of the adjacency matrix $A$. Thus, we obtain from (\ref{eq:interaction_model}) and the discrete-time approximation (\ref{eq:derivative_apx}) an approximate linear system as
\begin{equation} \label{eq:linear_system}
V_i \approx F_i  \begin{pmatrix}
a_{i1}\\
\vdots\\
a_{iN}
\end{pmatrix},
\end{equation}
where the $n\times 1$ vector $V_i$ equals
\begin{align}\label{eq:V_i_def}
V_i = \begin{pmatrix}
\dfrac{ x_i\left( \Delta t\right) - x_i\left( 0 \right) }{\Delta t} - f_i( 0 )\\
\vdots \\
\dfrac{ x_i\left( n \Delta t\right) - x_i\left( (n-1) \Delta t\right) }{\Delta t} - f_i( (n-1)\Delta t )
\end{pmatrix},
\end{align}
and the $n \times N$ matrix $F_i$ equals
\begin{align}\label{eq:F_i_def}
F_i = \begin{pmatrix}
g( x_i( 0 ), x_1( 0 ) )& ... & g( x_i( 0 ), x_N( 0 ) ) \\
\vdots & \ddots & \vdots \\
g( x_i( (n-1) \Delta t ), x_1( (n-1) \Delta t ) )& ... & g( x_i( (n-1) \Delta t ), x_N( (n-1) \Delta t ) )
\end{pmatrix}.
\end{align} 
From (\ref{eq:linear_system}), we obtain an estimate $\hat{A}$ of the adjacency matrix $A$ by solving 
\begin{align}\label{eq:lasso}
\begin{aligned} & \underset{a_{i1}, ..., a_{iN}}{\text{min}}  & & \left\lVert V_i - F_i  \begin{pmatrix}
a_{i1}\\
\vdots\\
a_{iN}
\end{pmatrix} \right\rVert^2_2 + \rho_i \sum^N_{j=1} a_{ij} & \\
 &\text{s.t.} & & a_{ij} \ge 0 \quad j=1, ..., N & 
\end{aligned} 
 \end{align}
 for every node $i$. The optimisation problem (\ref{eq:lasso}) is known as the \textit{least absolute shrinkage and selection operator} (LASSO) \cite{hastie2015statistical}. The application of LASSO, and variations thereof, to network reconstruction is an established approach \cite{shen2014reconstructing, timme2014revealing, wang2016data, prasse2019gemf}. The first addend in (\ref{eq:lasso}) measures the consistency of the link weights $a_{i1}$, ..., $a_{iN}$ with the observations $x(0), ..., x(n \Delta t)$, given the dynamical model (\ref{eq:interaction_model}). The second addend favours a sparse solution. The greater the \textit{regularisation parameter} $\rho_i>0$, the sparser the reconstructed adjacency matrix $\hat{A}$. We set the value of the parameter $\rho_i$ by hold-out cross-validation \cite{bergmeir2012use}. The LASSO (\ref{eq:lasso}) can be interpreted as Bayesian estimation problem, provided an exponential prior degree distribution of the adjacency matrix $A$. For more details on the reconstruction algorithm and the Bayesian interpretation, we refer the reader to Appendix~\ref{appendix:lasso}.

\section{The Prediction Accuracy versus the Reconstruction Accuracy}
To evaluate the prediction algorithm outlined in Section~\ref{sec:prediction_algorithm}, we consider the dynamics in Table~\ref{table:dynamical_models} on the respective real-world networks: \mbox{\textbf{(LV)} Food} web of Little Rock Lake \cite{martinez1991artifacts}, \mbox{\textbf{(MP)} Mutualistic} insect interactions \cite{kato1990insect,rezende2007non}, \mbox{\textbf{(MM)} Gene} regulatory network of the yeast S. Cerevisiae \cite{milo2002network}, \mbox{\textbf{(SIS)} Face}-to-face contacts between visitors of the \quotes{Infectious: stay away exhibition} \cite{isella2011s}, \mbox{\textbf{(KUR)} Structural} connectivity between brain regions \cite{van2013wu, tewarie2019spatially}, \mbox{\textbf{(CW)} C.} elegans neuronal connectivity \cite{white1986structure,chen2006wiring}. Appendix~\ref{appendix:empirical_networks_model_parameters} specifies the real-world networks and model parameters in detail.
  
\begin{figure}[!ht]
     \centering
     \begin{subfigure}[t]{0.3\textwidth}
         \centering
         \caption{LV ($N=183$)}
         \includegraphics[width=\textwidth]{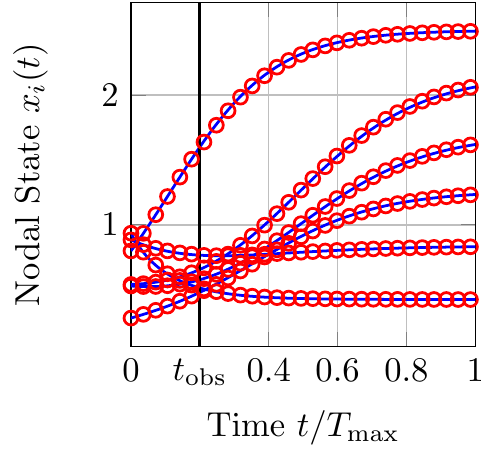}
         \label{fig:prediction_LV}
     \end{subfigure}
     \quad
     \begin{subfigure}[t]{0.3\textwidth}
         \centering
         \caption{MP ($N=679$)}
         \includegraphics[width=\textwidth]{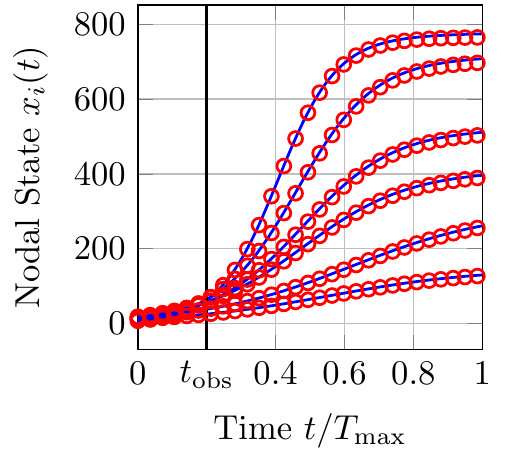}
         \label{fig:prediction_MP}
     \end{subfigure}
     \quad
     \begin{subfigure}[t]{0.3\textwidth}
         \centering
         \caption{MM ($N=620$)} 
         \includegraphics[width=\textwidth]{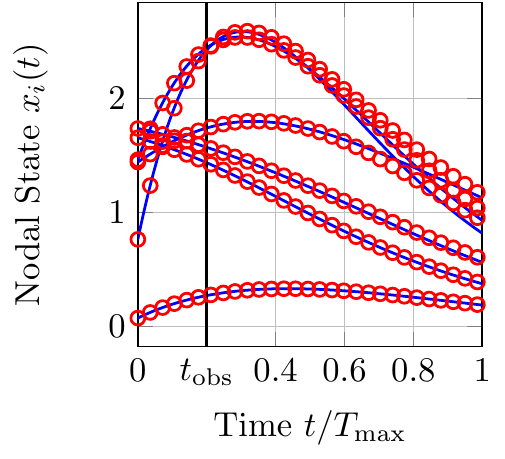}
         \label{fig:prediction_MM}
     \end{subfigure}
     \\ 
     %%% new row
     \begin{subfigure}[t]{0.3\textwidth}
         \centering
         \caption{SIS ($N=410$)}
         \includegraphics[width=\textwidth]{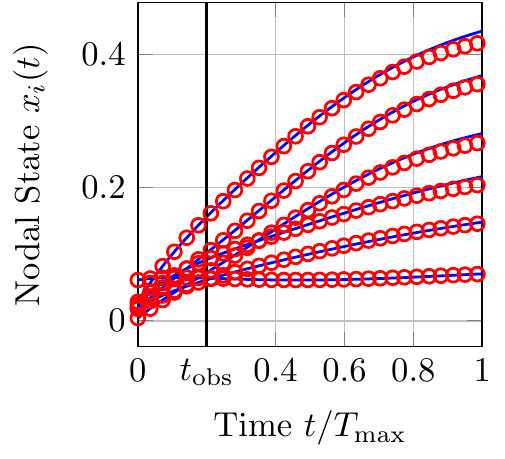}
         \label{fig:prediction_SIS}
     \end{subfigure}
     \quad
     \begin{subfigure}[t]{0.3\textwidth}
         \centering
         \caption{KUR ($N=78$)}
         \includegraphics[width=\textwidth]{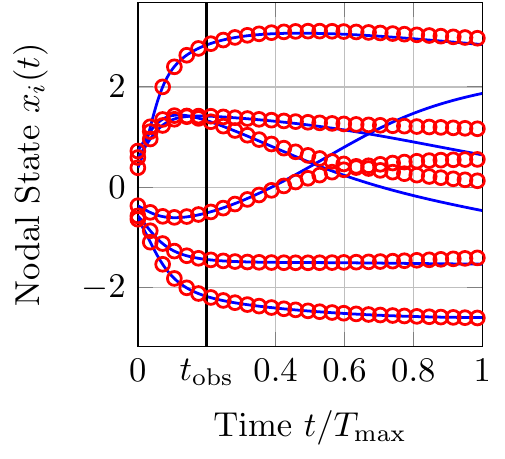}
         \label{fig:prediction_kuramoto}
     \end{subfigure}
     \quad
     \begin{subfigure}[t]{0.3\textwidth}
         \centering
         \caption{CW ($N=282$)}
         \includegraphics[width=\textwidth]{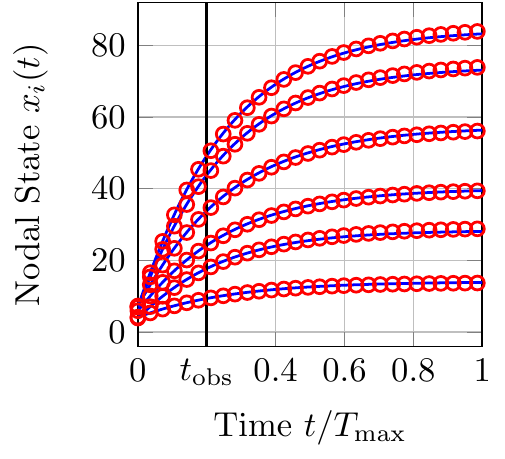}
         \label{fig:prediction_cw}
     \end{subfigure}     
        \caption{\textbf{Prediction Accuracy of Network Dynamics.} The blue curves are the true nodal states $x_i(t)$. The red marks are the nodal states $\hat{x}_i(t)$ on the reconstructed network $\hat{A}$, initialised at $\hat{x}(0)=x(0)$ and $\hat{x}(t_\textrm{obs})=x(t_\textrm{obs})$ for the time intervals $t< t_\textrm{obs}$ and $t \ge t_\textrm{obs}$, respectively. For readability, only six nodal states $x_i(t)$ are depicted for each network. The maximum prediction time~$T_\textrm{max}$ is different for each dynamic model, and the observation time equals $t_\textrm{obs} = T_\textrm{max}/5$. The number of observations is $n=100$.  \label{fig:dynamics_prediction}}        
\end{figure}  

 Figure~\ref{fig:dynamics_prediction} shows that the nodal state prediction $\hat{x}(t)$ is accurate at all times $t\ge t_\textrm{obs}$, except for the Kuramoto model. The Kuramoto nodal state prediction $\hat{x}(t)$ is accurate from time $t = t_\textrm{obs}$ until $t \approx 2 t_\textrm{obs}$, but then diverges from the true nodal state $x(t)$. In Section~\ref{sec:ill_conditioning}, we explain what makes the Kuramoto model different.

In view of the high prediction accuracy, we face the fundamental question: \textit{How similar are the topologies of the estimated network $\hat{A}$ and the true network $A$?} We quantify the similarity of the networks $A$ and $\hat{A}$ by two topological metrics. First, we consider the area under the receiver-operating-characteristic curve (AUC) \cite{fawcett2006introduction}. An AUC of 0.5 corresponds to reconstructing the network by tossing a coin for every possible link. The closer the AUC is to 1, the greater the similarity of the reconstructed topology to the true topology. Second, we consider the in-degree distribution of the matrices $A$ and $\hat{A}$. The in-degree $d_i$ of node $i$ equals the number of links that end at node $i$. The in-degree distribution is given by $\operatorname{Pr}\left[D \ge d\right]$, where $D$ is the degree of a randomly chosen node in the network.

\begin{figure}[!ht]
     \centering
     \begin{subfigure}[t]{0.3\textwidth}
         \centering
         \caption{LV: AUC=0.52}
         \includegraphics[width=\textwidth]{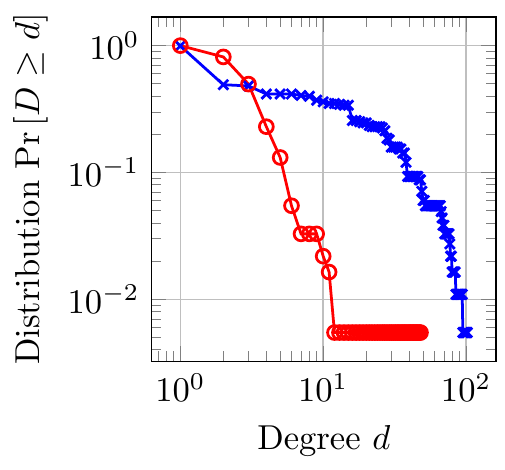}
         \label{fig:in_degrees_LV}
     \end{subfigure}
     \quad
     \begin{subfigure}[t]{0.3\textwidth}
         \centering
         \caption{MP: AUC=0.51}
         \includegraphics[width=\textwidth]{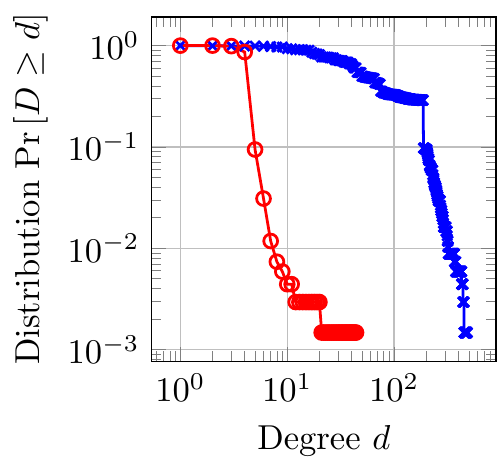}
         \label{fig:in_degrees_MP}
     \end{subfigure}
     \quad
     \begin{subfigure}[t]{0.3\textwidth}
         \centering
         \caption{MM: AUC=0.57}
         \includegraphics[width=\textwidth]{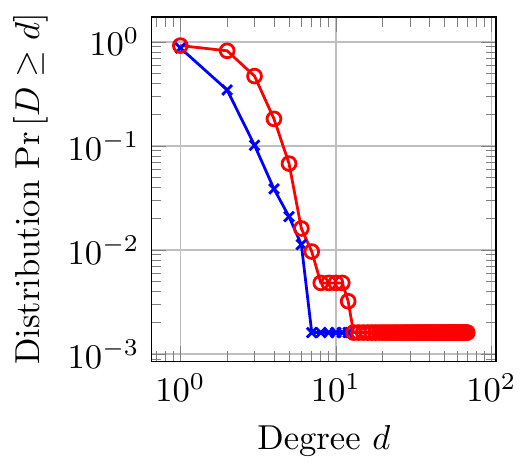}
         \label{fig:in_degrees_MM}
     \end{subfigure}
     \\
     %%%new row
     \begin{subfigure}[t]{0.3\textwidth}
         \centering
         \caption{SIS: AUC=0.52}
         \includegraphics[width=\textwidth]{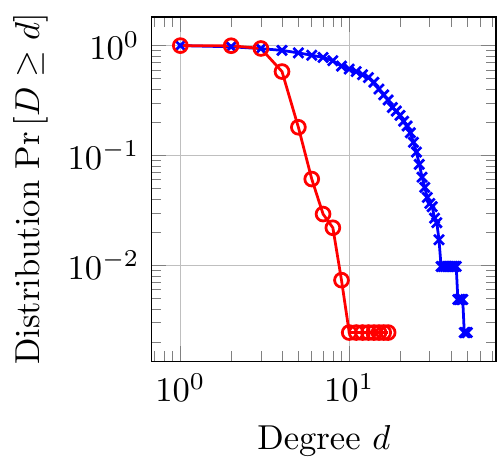}
         \label{fig:in_degrees_SIS}
     \end{subfigure}
     \quad
     \begin{subfigure}[t]{0.3\textwidth}
         \centering
         \caption{KUR: AUC=0.54}
         \includegraphics[width=\textwidth]{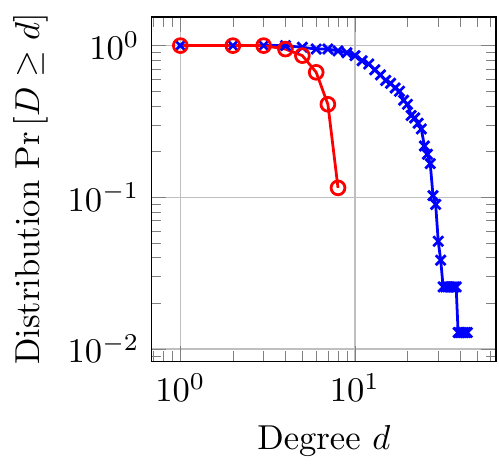}
         \label{fig:in_degrees_kuramoto}
     \end{subfigure}
     \quad
     \begin{subfigure}[t]{0.3\textwidth}
         \centering
         \caption{CW: AUC=0.53}
         \includegraphics[width=\textwidth]{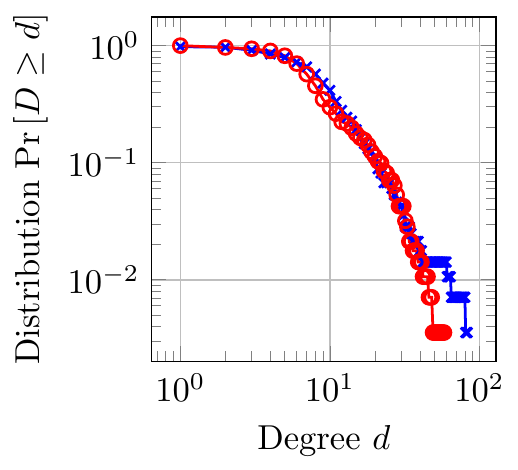}
         \label{fig:in_degrees_cw}
     \end{subfigure}     
        \caption{\textbf{Network Reconstruction Accuracy.} The reconstruction accuracy for the networks in Figure~\ref{fig:dynamics_prediction} with respect to two topological metrics. First, the AUC value of the reconstructed network $\hat{A}$. Second, the in-degree distributions $\operatorname{Pr}\left[D \ge d\right]$ for the estimated matrix $\hat{A}$ in red and the true matrix $A$ in blue.}
        \label{fig:in_degree_distribution}
\end{figure}     
     
Figure~\ref{fig:in_degree_distribution} compares the reconstructed network $\hat{A}$ to the true network $A$. The AUC value is close to 0.5 for all models. \textit{Hence, the topology of the reconstructed network bears practically no resemblance to the true network topology.} Moreover, the degree distribution $\operatorname{Pr}\left[D \ge d\right]$ of the reconstructed network differs strongly to the degree distribution of the true network, except for Figure~\ref{fig:in_degrees_MM} and Figure~\ref{fig:in_degrees_cw}. We emphasise that, even if two networks have the same degree distribution $\operatorname{Pr}\left[D \ge d\right]$, the network topologies can be entirely different. For instance, the AUC value equals only 0.53 in Figure~\ref{fig:in_degrees_cw}. 

\section{Proper Orthogonal Decomposition of the Nodal State Dynamics}\label{sec:ill_conditioning}

The dramatic contrast of accurate prediction of dynamics but inaccurate network reconstruction demands an explanation. The sole input to the prediction algorithm are the observations of the nodal state $x(t)$, which has two implications. First, the nodal state sequence $x(0), ..., x(n\Delta t)$ does not contain sufficient information to infer the network topology. Second, we do not need the topology to predict the nodal state $x(t)$. But, if not the topology, what else is required to accurately predict dynamics on networks? 

\begin{example}\label{example}
Consider a small network of $N=3$ nodes with the weighted adjacency matrix 
\begin{align*}
A = \begin{pmatrix}
0 & 0 & 2 \\
1 & 0 & 3 \\
0 & 1 & 1
\end{pmatrix}.
\end{align*}
Suppose that the nodal state vector equals $x(t)=( c_1(t), c_2(t), c_2(t))^T$ at every time $t$, where $c_1(t)$ and $c_2(t)$ denote some scalar functions. In other words, node 2 and node 3 have the same state at every time $t$. As vector equation, the nodal state $x(t)$ satisfies
\begin{align} \label{x_t_example}
x(t) = c_1(t) \begin{pmatrix}
1 \\
0 \\
0
\end{pmatrix}
+ c_2(t) \begin{pmatrix}
0 \\
1 \\
1
\end{pmatrix}.
\end{align}
For simplicity, we only consider the estimation of the links to node 1, i.e., $a_{11}$, $a_{12}$ and $a_{13}$. The evolution of the nodal state $x_1(t)$ follows from the dynamical model~(\ref{eq:interaction_model}) as
\begin{align*}
\frac{d x_1(t)}{dt} &= f_1(x_1(t)) + 2 g(x_1(t), x_3(t)).
\end{align*}
However, since $x_2(t)=x_3(t)=c_2(t)$ at every time $t$, it also holds that
\begin{align*} 
\frac{d x_1(t)}{dt} &= f_1(x_1(t)) + 2 g(x_1(t), x_2(t)).
\end{align*}
Thus, if we estimated the adjacency matrix $\hat{A}$ with $\hat{a}_{11}=0$, $\hat{a}_{12}=2$ and $\hat{a}_{13}=0$, then we could perfectly predict the nodal state $x_1(t)$. But neither estimate $\hat{a}_{12}$ nor $\hat{a}_{13}$ is equal to the true link weights $a_{12}$ and $a_{13}$, respectively. 
\end{example}
For Example \ref{example}, the estimate $\hat{A}$ yields a perfect prediction of the dynamics, because $2 g(x_1(t), x_3(t)) = 2 g(x_1(t), x_2(t))$. More generally, the estimated network $\hat{A}$ predicts the dynamics (\ref{eq:interaction_model}) exactly if and only if, at every future time $t\ge t_\textrm{obs}$,
\begin{align} \label{eq:A_Ahat_prediction_iff}
\sum^N_{j=1} \hat{a}_{ij} g\left(x_i(t), x_j(t)\right) = \sum^N_{j=1} a_{ij} g\left(x_i(t), x_j(t)\right).
\end{align} 
\textit{The network topology of the estimate $\hat{A}$ is relevant for predicting the dynamics only if the topology relates to (\ref{eq:A_Ahat_prediction_iff}).} We emphasise that (\ref{eq:A_Ahat_prediction_iff}) is linear with respect to the matrix $\hat{A}$ but not linear with respect to the nodal state~$x(t)$, unless the interaction function $g$ is linear.

In Example~\ref{example}, there exists a matrix $\hat{A} \neq A$ that satisfies (\ref{eq:A_Ahat_prediction_iff}), because the $3 \times 1$ nodal state vector $x(t)$ is equal to the linear combination (\ref{x_t_example}) of only 2 orthogonal vectors $y_1 = (1, 0, 0)^T$ and $y_2 = (0, 1, 1)^T$. In general, it is possible to \textit{approximate} any $N\times 1$ nodal state vector $x(t)$ by 
\begin{align}\label{x_t_low_rank}
x(t) \approx \sum^m_{p=1} c_p(t) y_{p}
\end{align}
at every time $t \in [0,T_\textrm{max}]$. Here, the \textit{agitation modes} $y_1, ..., y_m$ are some orthogonal vectors. The approximation (\ref{x_t_low_rank}) is known as \textit{proper orthogonal decomposition} \cite{antoulas2005approximation, kerschen2005method}. The more agitation modes~$m$, the more accurate the approximation (\ref{x_t_low_rank}). If $m=N$, then the approximation (\ref{x_t_low_rank}) is exact, because any $N\times 1$ vector $x(t)$ can be written as the linear combination of $N$ orthogonal vectors. Intuitively speaking, if the proper orthogonal decomposition (\ref{x_t_low_rank}) is accurate for $m<<N$ modes, then the nodal state vector $x(t)$ is barely agitated. 

In contrast to the zero-one vectors in Example \ref{example}, the agitation modes~$y_p$ are usually more complicated. We obtain the agitation modes~$y_p$ from the observations of the nodal state dynamics in two steps. First, we define the $N\times( n + 1)$ nodal state matrix as
\begin{align*}
X = \begin{pmatrix}
x(0) & x(\Delta t)& ... & x(n\Delta t)
\end{pmatrix}.
\end{align*}
Second, we obtain the agitation modes $y_1, ..., y_m$ as the first $m$ left-singular vectors of the nodal state matrix $X$. At any time $t\ge 0$, the scalar functions $c_p(t)$ follow as the inner product
\begin{align*}
c_p(t) = y^T_p x(t).
\end{align*}

Figure~\ref{fig:pod_from_lti} shows that the proper orthogonal decomposition (\ref{x_t_low_rank}) is accurate at all times $t\in [0, T_\textrm{max}]$. The number of agitation modes $y_p$ equals $m=15$, which is considerably lower than the network size~$N$. We emphasise that the agitation modes $y_p$ are computed from the nodal state $x(t)$ only until the observation time $t_\textrm{obs}$. Nevertheless, the proper orthogonal decomposition is accurate also at times $t\ge t_\textrm{obs}$. Hence, during the observation interval $[0, t_\textrm{obs}]$, the nodal state $x(t)$ quickly locks into only few agitation modes $y_p$, which govern the dynamics also for future time $t\ge t_\textrm{obs}$. For clarity, we stress that the proper orthogonal decomposition (\ref{x_t_low_rank}) cannot be used (directly) to predict the nodal state~$x(t)$: Additionally to the agitation modes $y_p$, the coefficients $c_p(t)$ at times $t\ge t_\textrm{obs}$ require the future nodal state $x(t)$.

\begin{figure}[!ht]
     \centering
     \begin{subfigure}[t]{0.3\textwidth}
         \centering
         \caption{LV ($N=183$)}
         \includegraphics[width=\textwidth]{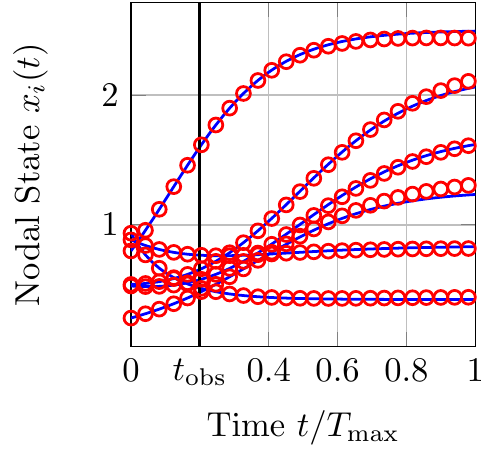}
         \label{fig:pod_LTI_LV}
     \end{subfigure}
     \quad
     \begin{subfigure}[t]{0.3\textwidth}
         \centering
         \caption{MP ($N=679$)}
         \includegraphics[width=\textwidth]{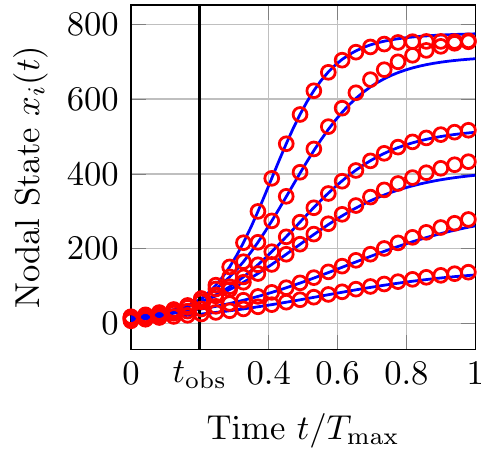}
         \label{fig:pod_LTI_MP}
     \end{subfigure}
     \quad
     \begin{subfigure}[t]{0.3\textwidth}
         \centering
         \caption{MM ($N=620$)}
         \includegraphics[width=\textwidth]{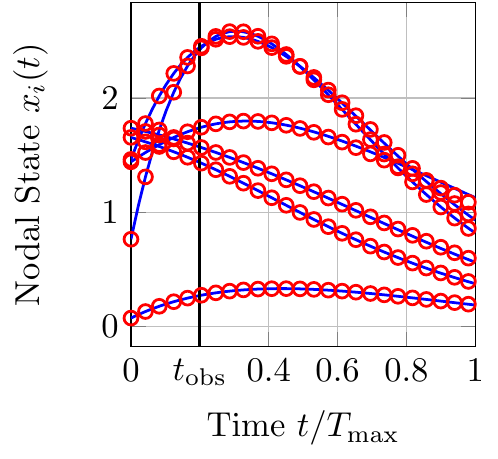}
         \label{fig:pod_LTI_MM}
     \end{subfigure}
     \\
     %%%new row
     \begin{subfigure}[t]{0.3\textwidth}
         \centering
         \caption{SIS ($N=410$)}
         \includegraphics[width=\textwidth]{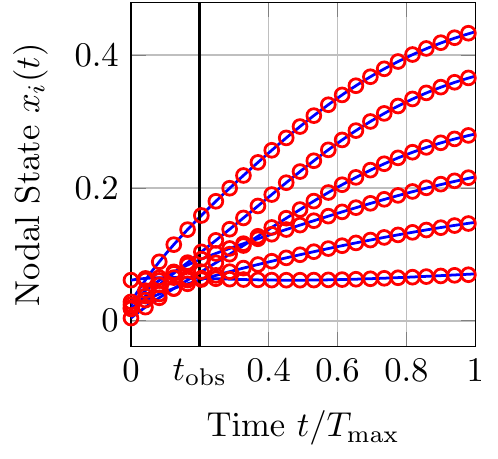}
         \label{fig:pod_LTI_SIS}
     \end{subfigure}
     \quad
     \begin{subfigure}[t]{0.3\textwidth}
         \centering
         \caption{KUR ($N=78$)}
         \includegraphics[width=\textwidth]{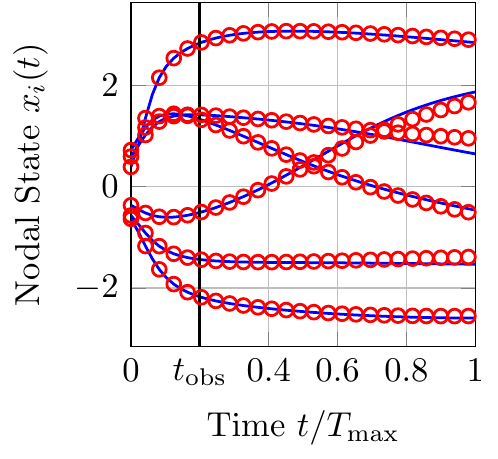}
         \label{fig:pod_LTI_kuramoto}
     \end{subfigure}
     \quad
     \begin{subfigure}[t]{0.3\textwidth}
         \centering
         \caption{CW ($N=282$)}
         \includegraphics[width=\textwidth]{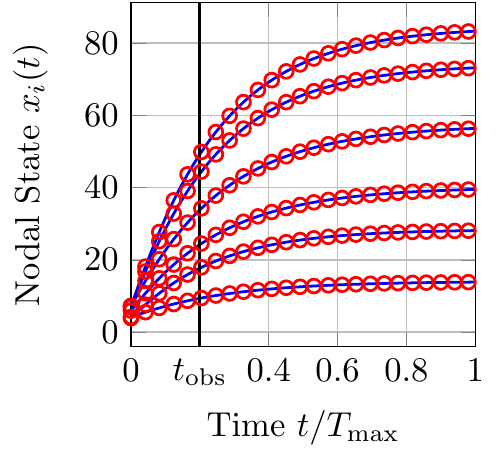}
         \label{fig:pod_LTI_cw}
     \end{subfigure}
    \caption{\textbf{Proper Orthogonal Decomposition of the Nodal State Dynamics.} The exact nodal state $x(t)$ in blue and the approximation (\ref{x_t_low_rank}) in red. For readability, only six nodal states $x_i(t)$ are depicted for each network. The approximation equals the linear combination of $m=15$ agitation modes $y_{1}$, ..., $y_{m}$, which are computed by observing the nodal state $x(t)$ from time $t=0$ to $t=t_\textrm{obs}$.} \label{fig:pod_from_lti}
\end{figure}

The Kuramoto oscillators are the only dynamics in Figure~\ref{fig:dynamics_prediction} that do not converge to a steady state~$x_\infty$. Hence, the proper orthogonal decomposition (\ref{x_t_low_rank}) is \textit{not} accurate when $t>>t_\textrm{obs}$, which explains that the prediction is least accurate for the Kuramoto model. 

Why, precisely, is it not possible to reconstruct the network $A$? The linear system (\ref{eq:linear_system}) forms the basis for the network reconstruction. The rank of the matrix $F_i$ is essential: If the matrix $F_i$ is of full rank, i.e., $\operatorname{rank}(F_i) = N$, then there is exactly one solution to (\ref{eq:linear_system}), namely the entries $a_{i1}, ..., a_{1N}$ of the true adjacency matrix~$A$. Otherwise, if $\operatorname{rank}(F_i) < N$, then there are infinitely many solutions to~(\ref{eq:linear_system}). If there is more than one solution to~(\ref{eq:linear_system}), then the LASSO estimation (\ref{eq:lasso}) results in the sparsest solution~$\hat{A}$ (with respect to the $\ell_1$-norm). 

We compute the numerical\footnote{Every computer works with finite precision arithmetic. Thus, not the exact rank but the \textit{numerical} rank of the matrix $F_i$ is decisive to solve the system (\ref{eq:linear_system}) in practice. The numerical rank equals the number of singular values of the matrix $F_i$ that are greater than a small threshold, which is set in accordance to the machine precision.} rank for Barab\'asi-Albert random graphs versus the network size $N$. Here, we consider the \textit{best case} for the network reconstruction: The derivative $dx_i(t)/dt$ is observed exactly, without any approximation error as in (\ref{eq:derivative_apx}). Hence, the system (\ref{eq:linear_system}) is satisfied with equality. Figure~\ref{fig:rank_vs_N} shows that the numerical rank of the matrix $F_i$ stagnates as the network size $N$ grows. Hence, the linear system (\ref{eq:linear_system}) is \textit{severely ill-conditioned} for large networks. For example, for the SIS process on a network with $N=1000$ nodes, we observe a $1000\times 1001$ nodal state sequence $x(0)$, ..., $x(1000\Delta t)$, but the numerical rank does not exceed~$32$.

The matrix $F_i$, defined by (\ref{eq:F_i_def}), follows from applying the nonlinear function $g$ to the nodal state~$x(t)$. The rank of the matrix $F_i$ is low, because the nodal state~$x(t)$ is barely agitated, see Appendix~\ref{appendix:ill_condition_matrix_F}.
               
 \begin{figure}[!ht]
         \centering
         \includegraphics[width=\textwidth]{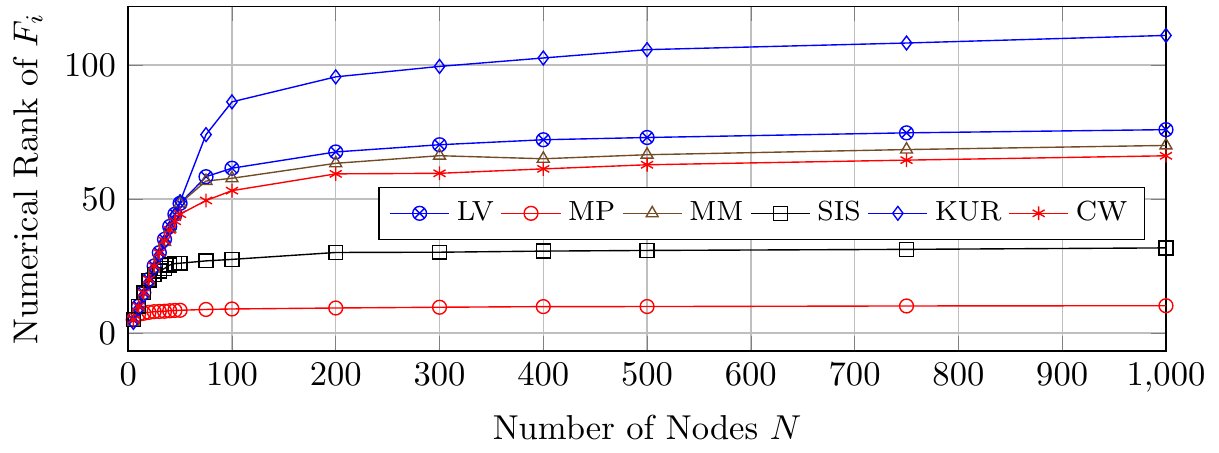}
         \caption{\textbf{Ill-Condition of the Network Reconstruction.} The numerical rank of the matrix~$F_i$ of the linear system (\ref{eq:linear_system}) versus the number of nodes $N$ for Barab\'asi-Albert random graphs. The observation time is set to $t_\textrm{obs}=T_\textrm{max}$, and the number of observations equals $n=1000$.}
         \label{fig:rank_vs_N}
     \end{figure}
                     
\section{Conclusions and Outlook}  

This works considers the prediction of general dynamics on networks, based on past observations of the dynamics. We proposed a prediction framework which consists of two steps. First, the network is estimated from the nodal state observations by the LASSO. The first step seemingly fails, since the estimated network bears no \textit{topological} similarity with the true network. Second, the nodal state is predicted by iterating the dynamical model on the inaccurately estimated network. Counterintuitively, the prediction is accurate!

The network reconstruction and prediction accuracy do not match, because the nodal state is barely agitated. Furthermore, the modes of agitation are hardly related to the network topology. Instead of the true topology, the estimated network does capture the interplay with the agitation modes.

We conclude with five points. First, the agitation modes depend on the initial nodal dynamics and, particularly, on the initial nodal state $x(0)$. As a result, the estimated network $\hat{A}$ depends on the initial state $x(0)$. Thus, as confirmed by numerical simulations, the adjacency matrix $\hat{A}$ may be useless for the prediction of dynamics with a different initial state $\tilde{x}(0)\neq x(0)$.
    
Second, the dynamics (\ref{eq:interaction_model}) are autonomous, since there is no control input. In some applications \cite{timme2007revealing, prasse2019gemf}, it may be possible to control the nodal state $x(t)$. Controlling the dynamics would result in more agitation modes of the nodal state $x(t)$. However, physical control constraints, such as power capacities, might limit the number of additional agitation modes.

Third, we could observe \textit{multiple} nodal state sequences $x(0), ..., x(n\Delta)$ with different initial states ~$x(0)$ on the same network. For sufficiently many sequences, we would observe enough agitation modes to reconstruct the network exactly, by stacking the respective linear systems (\ref{eq:linear_system}). However, Figure~\ref{fig:rank_vs_N} shows that the numerical rank of the matrix $F_i$ stagnates for large networks. Thus, the greater the network, the more time series must be observed to reconstruct the adjacency matrix~$A$. Observing a sufficiently great number of time series might not be viable, e.g., for the epidemic outbreak of a novel virus like SARS-CoV-2. 

Fourth, the proper orthogonal decomposition (\ref{x_t_example}) can be exact. If the network has equitable partitions, then the number of agitation modes equals the number of cells for some dynamical models \cite{van2010graph, o2013observability, bonaccorsi2015epidemic, schaub2016graph, devriendt2020bifurcation}. Furthermore, the SIS contagion dynamics reduce to only $m=1$ agitation mode around the epidemic threshold \cite{prasse2019time}.

Fifth, the dynamics on two different networks $\hat{A}\neq A$ is exactly the same only if the networks satisfy~(\ref{eq:A_Ahat_prediction_iff}). Based on numerical simulations, we showed that (\ref{eq:A_Ahat_prediction_iff}) holds approximately because the nodal state is barely agitated. We believe that the combination of the proper orthogonal decomposition (\ref{x_t_low_rank}) and equation (\ref{eq:A_Ahat_prediction_iff}) is a starting point for a further theoretical analysis on relating network structure and dynamics.

\section*{Acknowledgements}

We are grateful to Prejaas Tewarie for providing data on the structural brain network.

\appendix

\section{Network Reconstruction Algorithm: Details and Bayesian Interpretation}\label{appendix:lasso}
Subsection~\ref{subsec:pseudocode} states the reconstruction algorithm, which is an adaptation of the method we proposed in \cite{prasse2019gemf} for discrete-time epidemic models. The interpretation of the LASSO (\ref{eq:lasso}) as a Bayesian estimation problem is given in Subsection~\ref{subsec:bayesian}.

 \subsection{Details of the Network Reconstruction Algorithm} \label{subsec:pseudocode}
 
The solution $\hat{a}_{i1}(\rho_i), ..., \hat{a}_{iN}(\rho_i)$ to the LASSO (\ref{eq:lasso}) depends on the regularisation parameter $\rho_i$. We aim to choose the parameter $\rho_i$ that results in the solution $\hat{a}_{i1}(\rho_i), ..., \hat{a}_{iN}(\rho_i)$ with the greatest prediction accuracy. To assess the prediction accuracy, we apply \textit{hold-out cross-validation} \cite{bergmeir2012use}: We divide the nodal state observations into a training set $x(0),  ..., x(n_\textrm{train} \Delta t)$ and a validation set $x( (n_\textrm{train}+1) \Delta t), ..., x(n \Delta t)$. The training set is used to obtain the solution $\hat{a}_{i1}(\rho_i), ..., \hat{a}_{iN}(\rho_i)$ in dependency of $\rho_i$, whose prediction accuracy is evaluated on the validation set. We choose the regularisation parameter $\rho_i$ with the greatest prediction accuracy on the validation set.
    
   More precisely, we define the training set as the first $80\%$ of the nodal state observations $x(0)$, $x(\Delta t)$, ..., $x(n_\textrm{train} \Delta t)$, where $n_\textrm{train} = \lceil 0.8 n \rceil$. We denote the $n_\textrm{train}\times 1$ training vector $V_{\textrm{train},i}$ as 
\begin{align*}
V_{\textrm{train},i} = \begin{pmatrix}
\dfrac{ x_i\left( \Delta t\right) - x_i\left( 0 \right) }{\Delta t} - f_i( 0 )\\
\vdots \\
\dfrac{ x_i\left( n_\textrm{train} \Delta t\right) - x_i\left( \left(n_\textrm{train}-1\right) \Delta t\right) }{\Delta t} - f_i\left( \left(n_\textrm{train}-1\right)\Delta t \right)
\end{pmatrix} 
\end{align*}
 and the $n_\textrm{train} \times N$ training matrix $F_{\textrm{train},i}$ as
\begin{multline*}
F_{\textrm{train},i} =\\
 \begin{pmatrix}
g( x_i( 0 ), x_1( 0 ) )& ... & g( x_i( 0 ), x_N( 0 ) ) \\
\vdots & \ddots & \vdots \\
g\left( x_i\left( \left(n_\textrm{train}-1\right) \Delta t\right), x_1\left( \left(n_\textrm{train}-1\right) \Delta t\right) \right)& ... & g\left( x_i\left( \left(n_\textrm{train}-1\right) \Delta t\right), x_N\left( \left(n_\textrm{train}-1\right) \Delta t\right)\right)
\end{pmatrix}.
\end{multline*} 
We denote the solution of the LASSO (\ref{eq:lasso}) by $\hat{a}_{i1}(\rho_i), ..., \hat{a}_{iN}(\rho_i)$, when $F_i$ and $V_i$ are replaced by $F_{\textrm{train},i}$ and $V_{\textrm{train},i}$, respectively. If an entry $\hat{a}_{ij}(\rho_i)$ of the LASSO solution is smaller than the threshold $0.01$, then we round off and set $\hat{a}_{ij}(\rho_i)=0$. The prediction error $\operatorname{MSE}(\rho_i)$ of $\rho_i$ on the validation set is defined as 
 \begin{align}\label{eq:MSE_def}
 \operatorname{MSE}(\rho_i) = \left\lVert V_{\textrm{valid},i} - F_{\textrm{valid},i}  \begin{pmatrix}
\hat{a}_{i1}(\rho_i)\\
\vdots\\
\hat{a}_{iN}(\rho_i)
\end{pmatrix} \right\rVert^2_2.
 \end{align}
Here, the $(n - n_\textrm{train})\times 1$ validation vector $V_{\textrm{valid},i}$ and the $(n - n_\textrm{train}) \times N$ validation matrix $F_{\textrm{valid},i}$ are defined by the nodal state observations $x(( n_\textrm{train} +1 )\Delta t), ..., x(n \Delta t)$, analogously to the training vector $V_{\textrm{train},i}$ and the training matrix $F_{\textrm{train},i}$

We iterate over a set $\Theta_i$, specified below, of predefined candidate values for $\rho_i$. Every candidate value $\rho_i \in \Theta_i$ results in a different prediction error $\operatorname{MSE}(\rho_i)$. We determine the final regularisation parameter $\rho_{\textrm{opt}, i}$ as the candidate value $\rho_{\textrm{opt}, i} \in \Theta_i$ with the minimal prediction error $\operatorname{MSE}(\rho_{\textrm{opt}, i})$. We obtain the final estimate $\hat{a}_{i1}, ..., \hat{a}_{iN}$ as the solution to the LASSO (\ref{eq:lasso}) with the regularisation parameter $\rho_{\textrm{opt}, i}$, using the matrix $F_i$ and vector $V_i$ from \textit{all} nodal state observations $x(0), ..., x(n\Delta t)$.
 
 We define the set $\Theta_i$ as 20 logarithmically equidistant candidate values as $\Theta_i=\{\rho_{\textrm{min}, i}, ..., \rho_{\textrm{max}, i}\}$. If $\rho_i> \rho_{\textrm{th}, i}$, where $\rho_{\textrm{th}, i} = 2 \lVert F^T_i V_i \rVert_\infty$, then \cite{kim2007interior} the solution to the LASSO (\ref{eq:lasso}) equals $a_{ij}=0$ for all nodes~$j$. Thus, we set the candidate values in the set $\Theta_i$ proportional to $\rho_{\textrm{th}, i}$. We define $\rho_{\textrm{min}, i}= 2 \lVert F^T_i V_i \rVert_\infty 10^{-6}$ and $\rho_{\textrm{max}, i}= 2 \lVert F^T_i V_i \rVert_\infty 10^{-2}$. The network reconstruction method is given by Algorithm~\ref{algorithm:reconstruction}.

\begin{algorithm}
	\caption{\texttt{Network reconstruction}}
	\begin{algorithmic}[1]
		\State \textbf{Input: }  nodal state time series $x(0), x(\Delta t), ..., x(n\Delta t)$
		\State \textbf{Output: } estimated adjacency matrix $\hat{A}$ with the elements $\hat{a}_{ij}$
		\For {$i=1, ..., N$}			
		\State $\rho_{\textrm{max}, i}\gets 2 \lVert F^T_i V_i \rVert_\infty 10^{-2}$
		\State $\rho_{\textrm{min}, i}\gets  ~10^{-4}\rho_{\textrm{max}, i}$
		\State $\Theta_i \gets$ 20 logarithmically equidistant values from $\rho_{\textrm{min}, i}$ to $\rho_{\textrm{max}, i}$				
		\For {$\rho_i \in \Theta_i$}			
		\State $\hat{a}_{i1}(\rho_i), ..., \hat{a}_{iN}(\rho_i) \gets$ solution to (\ref{eq:lasso}) on the training set $V_{\textrm{train},i}$ and $F_{\textrm{train},i}$		
		\State $\hat{a}_{ij}(\rho_i)\gets 0$ for all $\hat{a}_{i1}(\rho_i), ..., \hat{a}_{iN}(\rho_i)$ smaller than $0.01$	
		\State Compute $\operatorname{MSE}(\rho_i)$ by (\ref{eq:MSE_def}) on the validation set $V_{\textrm{valid},i}$ and $F_{\textrm{valid},i}$		
		\EndFor
		\State $\rho_{\textrm{opt}, i} \gets \underset{\rho_i \in \Theta_i}{\operatorname{arg min}}\operatorname{MSE}\left(\rho_i\right)$
	\State $(\hat{a}_{i1}, ..., \hat{a}_{iN} )\gets$ the solution to (\ref{eq:lasso}) for $\rho_i = \rho_{\textrm{opt}, i}$ on the whole data set $F_i, V_i$ 
	\State $\hat{a}_{ij}\gets 0$ for all $\hat{a}_{i1}, ..., \hat{a}_{iN}$ smaller than $0.01$	
	\EndFor
	\end{algorithmic}
	\label{algorithm:reconstruction}
\end{algorithm}

 \subsection{Interpretation as a Bayesian Estimation}\label{subsec:bayesian}

For every node $i$, we define the error $w_i( k\Delta t )$ of the first-order approximation (\ref{eq:derivative_apx}) of the derivative $d x_i(t)/dt$ at time $t=k\Delta t$, such that  
\begin{equation}\label{eq:apx_error_def}
x_i\left( (k+1) \Delta t\right) = x_i\left( k \Delta t\right) + \Delta t \left.\frac{dx_i(t)}{dt}\right|_{t = k\Delta t} + w_i( k\Delta t ).
\end{equation}
The approximation (\ref{eq:apx_error_def}) can be regarded as a nonlinear system in discrete time $k$ with random model errors $w_i( k\Delta t )$, which forms the basis for the Bayesian interpretation of the LASSO (\ref{eq:lasso}). Furthermore, we rely on two assumptions.

\begin{assumption} \label{assumption:gaussian_model_errors}
For every node $i$ at every time $k$, the approximation error $w_i( k\Delta t )$ follows the normal distribution $\mathcal{N}\left(0, \sigma^2_w\right)$ with zero mean and variance $\sigma^2_w$. Furthermore, the approximation errors $w_i( k\Delta t )$ are stochastically independent and identically distributed at all times $k=1, ..., n$ and for all nodes $i$.
\end{assumption}
The exact model error $w_i( k\Delta t )$ is difficult to analyse, since $w_i( k\Delta t )$ is determined by higher-order derivatives of the nodal state $x(t)$. In contrast, assuming that the model error $w_i( k\Delta t )$ follows a Gaussian distribution $\mathcal{N}\left(0, \sigma^2_w\right)$ allows for a simple analysis. Furthermore, Assumption~\ref{assumption:gaussian_model_errors} stems from the \textit{maximum entropy principle} \cite{jaynes1957information}: Given a set of constraints on a probability distribution (e.g., specified mean), assume the \quotes{least informative} distribution, i.e., the distribution with maximum entropy that satisfies those constraint. Among all distributions on $\mathbb{R}$ with zero mean and variance $\sigma^2_w$, the Gaussian distribution $\mathcal{N}\left(0, \sigma^2_w\right)$ has the maximum entropy \cite{papoulis2002probability}.

\begin{assumption} \label{assumption:exponential_degree_distribution}
The adjacency matrix $A$ with non-negative elements $a_{ij}\ge 0$ follows the prior distribution
\begin{align}\label{eq:exp_distribution}
\operatorname{Pr} \left[ A \right] = \alpha \exp\left( -\sum^N_{i=1}  \sum^N_{j=1} a_{ij} \right),
\end{align}
where the normalisation constant $\alpha$ is set such that
\begin{align*}
 \int_{\mathbb{R}^{N\times N}_{\ge 0}} \operatorname{Pr} \left[ A \right] dA = 1.
\end{align*}
Furthermore, the matrix $A$ and the initial nodal state $x(0)$ are stochastically independent.
\end{assumption}
Clearly, there are more suitable random graph models for real-world networks than the exponential degree distribution in Assumption~\ref{assumption:exponential_degree_distribution}. In particular, the degree distribution of many real-world networks follows a power-law \cite{barabasi2016network}. However, the central result in this work is given by the juxtaposition of Figure~\ref{fig:dynamics_prediction} with Figure~\ref{fig:in_degree_distribution}: It is not necessary to accurately reconstruct the degree distribution to predict the dynamics on a network. Thus, even with the potentially imprecise assumption on the degree distribution (\ref{eq:exp_distribution}), it is possible to accurately predict the dynamics on the network. 

Proposition~\ref{proposition:bayesian_estimation} states the Bayesian interpretation of the LASSO (\ref{eq:lasso}). We emphasise that Proposition~\ref{proposition:bayesian_estimation} is not novel and follows standard arguments in parameter estimation, see for instance \cite{aster2018parameter}. Furthermore, Tibshirani elaborated on the Bayesian interpretation of the LASSO in the seminal paper \cite{tibshirani1996regression}. Nevertheless, we believe that the presentation of Proposition~\ref{proposition:bayesian_estimation}, here in the context of network reconstruction, is valuable to the reader.

\begin{proposition} \label{proposition:bayesian_estimation}
Suppose that Assumption~\ref{assumption:gaussian_model_errors} and Assumption~\ref{assumption:exponential_degree_distribution} hold true and that the nodal state $x(t)$ follows (\ref{eq:apx_error_def}). Then, provided the regularisation parameter equals $\rho_i=2 \sigma^2_w/\Delta t^2$, the matrix $\hat{A}$, which is obtained by solving the LASSO (\ref{eq:lasso}) for every node~$i$, coincides with the Bayesian estimate:
\begin{align}\label{eq:MAP_estimation}
\hat{A} = ~ \underset{A}{\operatorname{arg max}} ~ \operatorname{Pr} \left[ A \big| x(0), ..., x(n\Delta t) \right].
\end{align}
\end{proposition}
\begin{proof}
Analogous steps to the derivations in \cite{prasse2018exact} yield that (\ref{eq:MAP_estimation}) is equivalent to
\begin{align*}
\hat{A} = ~ \underset{A}{\operatorname{arg max}} ~ \log\left( \operatorname{Pr} \left[ A \right] \right)+\sum^{n-1}_{k=0} \log\left( \operatorname{Pr} \left[ x((k+1)\Delta t)  \big| x( k\Delta t), A\right]\right).
\end{align*}
Under Assumption~\ref{assumption:gaussian_model_errors}, the errors $w_i( k\Delta t )$ are independent for different nodes $i$. Thus, we obtain that 
\begin{align}\label{eq:MAP_derivation_1}
\hat{A} = ~ \underset{A}{\operatorname{arg max}} ~ \log\left( \operatorname{Pr} \left[ A \right] \right)+\sum^{n-1}_{k=0} \sum^N_{i=1} \log\left( \operatorname{Pr} \left[ x_i((k+1)\Delta t)  \big| x( k\Delta t), A\right]\right).
\end{align}
The probability $\operatorname{Pr} \left[ x_i((k+1)\Delta t)  \big| x( k\Delta t), A\right]$ is determined by the distribution of the error $w_i( k\Delta t )$. From (\ref{eq:interaction_model}) and (\ref{eq:apx_error_def}), it follows that 
\begin{align*}
w_i( k\Delta t ) = x_i((k+1)\Delta t) - x_i( k\Delta t ) - \Delta t \left( f_i\left( x_i(k\Delta t) \right) + \sum^N_{j=1} a_{ij} g\left(x_i(k\Delta t), x_j(k\Delta t)\right)\right).
\end{align*}
With the definition of the vector $V_i$ and the matrix $F_i$ in (\ref{eq:V_i_def}) and (\ref{eq:F_i_def}), respectively, we obtain that
\begin{align*}
w_i( (k-1)\Delta t ) = \Delta t \left( V_i \right)_k - \Delta t \sum^N_{j=1} \left( F_i\right)_{kj} a_{ij}.
\end{align*}
Hence, under Assumption~\ref{assumption:exponential_degree_distribution} on the prior $\operatorname{Pr} \left[ A \right]$, the optimisation problem (\ref{eq:MAP_derivation_1}) becomes
\begin{align*}
 \begin{aligned}\hat{A} = ~ &\underset{A}{\operatorname{arg max}} ~ & & \log\left(\alpha \right) - \sum^N_{i=1}  \sum^N_{j=1} a_{ij} &\\
&&&+\sum^N_{i=1} \sum^{n}_{k=1}  \log\left( \operatorname{Pr} \left[ w_i( (k-1)\Delta t ) = \Delta t \left( V_i \right)_k  - \Delta t \sum^N_{j=1} \left( F_i\right)_{kj} a_{ij} \right]\right)& \\
 &\text{s.t.} & & a_{ij} \ge 0 \quad i,j=1, ..., N. &
\end{aligned} 
\end{align*}
The term $\log\left(\alpha \right)$ is constant with respect to the matrix $A$ and can be omitted. Furthermore, the optimisation can be carried out independently for every node $i$, which yields that
\begin{align*}
 \begin{aligned}&\underset{a_{i1}, ..., a_{iN}}{\operatorname{max}} ~ & &\sum^{n}_{k=1}  \log\left( \operatorname{Pr} \left[ w_i( (k-1)\Delta t ) = \Delta t \left( V_i \right)_k  - \Delta t \sum^N_{j=1} \left( F_i\right)_{kj} a_{ij} \right]\right) - \sum^N_{j=1} a_{ij} &\\
 &\text{s.t.} & & a_{ij} \ge 0 \quad j=1, ..., N. &
\end{aligned} 
\end{align*}
Under Assumption~\ref{assumption:gaussian_model_errors}, the errors $w_i( k\Delta t )$ follow a Gaussian distribution, which results in the minimisation problem 
\begin{align*}
 \begin{aligned}&\underset{a_{i1}, ..., a_{iN}}{\operatorname{min}} ~ & & \sum^{n}_{k=1}\log(\sqrt{2\pi} \sigma_w)+ \frac{1}{2 \sigma^2_w}
\left( \Delta t \left( V_i \right)_k - \Delta t \sum^N_{j=1} \left( F_i\right)_{kj} a_{ij}  \right)^2 + \sum^N_{j=1} a_{ij} &\\
 &\text{s.t.} & & a_{ij} \ge 0 \quad j=1, ..., N. &
\end{aligned}
\end{align*}
Omitting the constant term $\log(\sqrt{2\pi} \sigma_w)$ and multiplying with $2\sigma^2_w/\Delta t^2$ gives 
\begin{align*}
\begin{aligned}&\underset{a_{i1}, ..., a_{iN}}{\operatorname{min}} ~ & & \sum^{n}_{k=1}
\left(\left( V_i \right)_k - \sum^N_{j=1} \left( F_i\right)_{kj} a_{ij} \right)^2 + 2 \frac{\sigma^2_w}{\Delta t^2}\sum^N_{j=1} a_{ij} &\\
 &\text{s.t.} & & a_{ij} \ge 0 \quad j=1, ..., N. &
\end{aligned}
\end{align*}
By identifying $\rho_i=2 \sigma^2_w/\Delta t^2$, we obtain the LASSO (\ref{eq:lasso}), which completes the proof.
\end{proof}

\section{Details on the Empirical Networks and Model Parameters}\label{appendix:empirical_networks_model_parameters}
Here, we provide details on the empirical networks and the parameters for the respective network dynamics in Section~\ref{sec:dynmical_models}. For every network topology, we obtain the link weights $a_{ij}$ as follows. If there is a link from node $j$ to node $i$, then we set the element $a_{ij}$ to a uniformly distributed random number in $[0.5,1.5]$. If there is no link from node $j$ to node $i$, then we set the respective element to $a_{ij}=0$.

\subsection{Lotka-Volterra Population Dynamics} 
 For the competitive population dynamics described by the Lotka-Volterra equations, we consider the \textit{Little Rock Lake} network \cite{martinez1991artifacts}, which we accessed via the \textit{Konect} network collection \cite{kunegis2013konect}. The asymmetric and connected network consists of $N=183$ nodes, which correspond to different species. There are $L = 2494$ directed links which specify the predation of one species upon another. 
 
  For every species $i$, we set the growth parameters $\alpha_i$ and $\theta_i$ to uniformly distributed random numbers in $[0.5,1.5]$. Furthermore, we set the the initial nodal state $x_i(0)$ to a uniformly distributed random number in $[0,1]$ for every species $i$. We set the maximum prediction time to $T_\textrm{max}=5$.
        
\subsection{Mutualistic Population Dynamics}
Kato \textit{et al.} \cite{kato1990insect} studied the relationship between $679$ insect species and $91$ plants in a beech forest in Kyoto by specifying which insects pollinate or disperse which plant. We accessed the insect-plant network via the supplementary data in \cite{rezende2007non}. The insect-plant network determines a mutualistic insect-insect network \cite{harush2017dynamic}: If two insect species $i$ and $j$ pollinate or disperse the same plant, then both insect species $i$ and $j$ contribute to, and benefit from, the abundance of the plant. Thus, if two insect species~$i$ and $j$ are linked to the same plant, then we set $a_{ij}$ to a uniformly distributed random number in $[0.5,1.5]$, and $a_{ij}=0$ otherwise. As a result, we obtain a symmetric and disconnected network with $N=679$ nodes and $L = 30,905$ links.

  For every species $i$, we set the growth parameters $\alpha_i$ and $\theta_i$ to a uniformly distributed random number in $[0.5,1.5]$. Furthermore, we set the the initial nodal state $x_i(0)$ to a uniformly distributed random number in $[0,20]$ for every species $i$. We set the maximum prediction time to $T_\textrm{max}=0.025$.

\subsection{Michaelis-Menten Regulatory Dynamics}
We consider the transcription interactions between regulatory genes in the yeast \textit{S. Cerevisiae} \cite{milo2002network}. The asymmetric and disconnected network has $N = 620$ and $L = 869$ links. The influence from gene~$j$ to gene $i$ is in either an activation or inhibition regulation. Since the activator interactions account for more than $80\%$ of the links between genes, we only consider activation interactions, see also \cite{barzel2013universality}. In line with Harush and Barzel \cite{harush2017dynamic}, we consider degree avert regulatory dynamics by setting the Hill coefficient to $h=2$. We set the the initial nodal state $x_i(0)$ to a uniformly distributed random number in $[0,2]$ for every node $i$. We set the maximum prediction time to $T_\textrm{max}=3$.

\subsection{Susceptible-Infected-Susceptible Epidemics}
The SIS contagion dynamics are evaluated on the contact network of the \textit{Infectious: Stay Away} exhibition \cite{isella2011s} between $N=410$ individuals, accessed via \cite{kunegis2013konect}. The connected and symmetric network has $L = 5530$ links. A link between two nodes $i$,$j$ indicates that the respective two individuals had a face-to-face contact that lasted for at least 20 seconds.

A crucial quantity for the SIS dynamics is the basic reproduction number $R_0$, which is defined as~\cite{van2002reproduction} 
\begin{align}\label{eq:R_0_def}
R_0 = \rho\left( \operatorname{diag}\left(\delta_1, ..., \delta_N\right)^{-1} B \right).
\end{align}
Here, the spectral radius of an $N \times N$ matrix $M$ is denoted by $\rho(M)$, and $\operatorname{diag}\left(\delta_1, ..., \delta_N\right)$ denotes the $N\times N$ diagonal matrix with the curing rates $\delta_1, ..., \delta_N$ on its diagonal. If the basic reproduction number $R_0$ is less than or equal to 1, then the epidemic dies out \cite{lajmanovich1976deterministic}, i.e., $x(t)\rightarrow 0$ as $t\rightarrow \infty$. We would like to study the spread of a virus that does not die out, and we aim to set the basic reproduction number to $R_0=1.5$: First, we set the \quotes{initial curing rate} $\delta^{(0)}_i$ to a uniformly distributed random number in $[0.5, 1.5]$ for every node $i$. Then, we set the curing rates to $\delta_i=c \delta^{(0)}_i$, where the multiplicity~$c$ is chosen such that the basic reproduction number in (\ref{eq:R_0_def}) equals $R_0=1.5$. We set the the initial nodal state $x_i(0)$ to a uniformly distributed random number in $[0,0.1]$ for every node $i$. Furthermore, we set the maximum prediction time to $T_\textrm{max}=0.5$.

\subsection{Kuramoto Oscillators}

We consider Kuramoto oscillator dynamics on the structural human brain network \cite{breakspear2010generative} of size $N = 78$. Every node corresponds to a brain region of the automated anatomical labelling (AAL) atlas \cite{tzourio2002automated}. The structural brain network specifies the anatomical connectivity between regions, i.e., the physical connections between regions based on white matter tracts. White matter tracts were estimated using fibre tracking from diffusion MRI data from the \textit{Human Connectome Project} \cite{van2013wu} as outlined in \cite{tewarie2019spatially}. The network is symmetric and has $L = 696$ links.

For every node $i$, we set the natural frequency $\omega_i$ to a normally distributed random number with zero mean and standard deviation $0.1 \pi$. Furthermore, we set the the initial nodal state $x_i(0)$ to a uniformly distributed random number in $[-\pi/4, \pi/4]$ for every node $i$. We set the maximum prediction time to $T_\textrm{max}=1$.

\subsection{Cowan-Wilson Neural Firing}
We consider the modified Cowan-Wilson neural firing model of Laurence \textit{et al.} \cite{laurence2019spectral} on the neuronal connectivity of the adult \textit{Caenorhabditis elegans} hermaphrodite worm. Originally, White \textit{et al.} \cite{white1986structure} compiled the neuronal connectivity of C. elegans. In \cite{chen2006wiring, varshney2011structural}, the neural wiring was updated, which we accessed online via the \textit{Wormatlas} online database\footnote{Under the link: http://www.wormatlas.org/neuronalwiring.html}. The somatic nervous system has $N=282$ neurons and $L = 2994$ synapses. A link from node $j$ to node $i$ indicates the presence of at least one synapse from neuron $j$ to neuron $i$. 

The slope and the threshold of the neural activation functions are set to $\tau =1$ and $\mu=1$, respectively. The initial state $x_i(0)$ of every node $i$ is set to a uniformly distributed random number in $[0,10]$. We set the maximum prediction time to $T_\textrm{max}=4$.

\section{Ill-Conditioning of the Network Reconstruction} \label{appendix:ill_condition_matrix_F}

We argue that if the proper orthogonal decomposition (\ref{x_t_low_rank}) is accurate, then the matrix $F_i$ in (\ref{eq:F_i_def}) is ill-conditioned. We rewrite the matrix $F_i$ as
\begin{align*}
F_i = \begin{pmatrix}
r^T_i(x\left(0\right))\\
\vdots \\
r^T_i(x\left((n-1)\Delta t\right))
\end{pmatrix},
\end{align*}
where the rows are given by $N\times 1$ vectors
\begin{align}\label{r_function_def}
r_i\left( x\left( k \Delta t \right) \right) = \begin{pmatrix}
g\left(x_i\left( k \Delta t \right), x_1\left( k \Delta t \right)\right) \\
\vdots \\
g\left(x_i\left( k \Delta t \right), x_N\left( k \Delta t \right) \right) 
\end{pmatrix}.
\end{align}
We aim to show that the approximation of the nodal state $x(t)$ in (\ref{x_t_low_rank}) implies that the row vectors $r_i(x(k\Delta t))$, where $k=0, 1, ..., n-1$, can be approximated by the linear combination of only few vectors, which implies the ill-conditioning of the matrix $F_i$. To shorten the notation, we drop the time index $t$ in this section. More precisely, we formally replace the nodal state $x(t)$ by $x$ and the functions $g\left( x_i(t), x_j(t)\right)$ and $r_i(x(t))$ by $g\left( x_i, x_j \right)$ and $r_i(x)$, respectively. 

To analyse the nonlinear dependency of the rows $r_i(x)$ on the nodal state $x$, we resort to a Taylor expansion of the rows $r_i(x)$. The function $r_i:\mathbb{R}^N \rightarrow \mathbb{R}^N$ is specified by the nonlinear interaction function $g$ of the dynamical model (\ref{eq:interaction_model}). The Taylor expansion of the function $g(x_i,x_j)$ around the point $x_i=x_j=0$ reads
\begin{align}\label{oujnkjbsdfs}
g\left( x_i, x_j \right)= g(0, 0) + \sum^\infty_{k=1} \sum_{\alpha+\beta = k} \frac{1}{\alpha! \beta!} x^\alpha_i x^\beta_j \left. \frac{\partial^k g \left( x_i, x_j \right)}{\partial x^\alpha_i \partial x^\beta_j}\right|_{x_i=x_j=0} . 
\end{align}
We define the coefficient $\eta\left( \alpha, \beta \right)$ as
\begin{align}\label{eta_def}
\eta\left( \alpha, \beta \right) = \frac{1}{\alpha! \beta!} \left.\frac{\partial^k g \left( x_i, x_j \right)}{\partial x^\alpha_i \partial x^\beta_j}\right|_{x_i= x_j=0}.
\end{align}
The indices $i$ and $j$ refer to the first and second argument of the function $g(x_i, x_j)$. Thus, the coefficient $\eta\left( \alpha, \beta \right)$ does not depend on the value of the indices $i$, $j$. With (\ref{eta_def}), it follows from (\ref{oujnkjbsdfs}) that
\begin{align} \label{sefsdfsdfsdf}
g\left( x_i, x_j \right)= g(0, 0) + \sum^\infty_{k=1} \sum^k_{\alpha=0} \eta\left( \alpha, k-\alpha \right)  x^\alpha_i x^{k-\alpha}_j . 
\end{align}
With (\ref{sefsdfsdfsdf}), we obtain the Taylor series of the function $r_i$ in (\ref{r_function_def}) as
\begin{align} \label{r_taylor}
r_i(x) = r_i(0) + \sum^\infty_{k=1} \sum^k_{\alpha=0} \eta\left( \alpha, k-\alpha \right)  x^\alpha_i x^{k-\alpha},
\end{align}
where we denote the element-wise power of the vector $x$ as
\begin{align*} 
x^{k-\alpha} = \left( x^{k-\alpha}_1, ..., x^{k-\alpha}_N \right)^T.
\end{align*}
We define the truncation of the series (\ref{r_taylor}) until the $q$-th power as
\begin{align}\label{r_qi_def}
r_{q, i}\left( x \right) = r_i(0) + \sum^{q}_{k=1} \sum^k_{\alpha=0} \eta\left( \alpha, k-\alpha \right)  x^\alpha_i x^{k-\alpha}.
\end{align}
For a sufficiently great power $q$, the matrix $F_i$ is approximated by $F_i\approx F_{q, i}$, where we define the matrix $F_{q, i}$ with the truncations~$r_{q,i}(x)$ as 
\begin{align*}
F_{q, i} = \begin{pmatrix}
r^T_{q,i}(x\left(0\right))\\
\vdots \\
r^T_{q,i}(x\left((n-1)\Delta t\right))
\end{pmatrix}.
\end{align*}
In fact, if the interaction function $g(x_i, x_j)$ is a polynomial of degree $q$, then the matrices $F_i$ and $F_{q,i}$ coincide, i.e., $F_i=F_{q, i}$. For instance, it holds that $F_i = F_{2, i}$ for the SIS epidemic process whose interaction function equals $g(x_i,x_j)=(1-x_i)x_j$. The low agitation of the nodal state vector $x$ in (\ref{x_t_low_rank}) indeed explains the ill-conditioning of the matrix $F_i$:
\begin{proposition}\label{proposition:rank}
Suppose that the $N\times 1$ nodal state vector $x(t)$ equals the linear combination of $m$ vectors~$y_{p}$ at every time $t$, 
\begin{align} \label{kjnkdsd}
x(t) = \sum^m_{p=1}c_{p}(t)y_{p}
\end{align}
for some scalars $c_{p}(t) \in \mathbb{R}$. Then, the rank of the matrix $F_{q,i}$ is bounded by
\begin{align*} 
\operatorname{rank}\left( F_{q,i} \right) \le\sum^q_{\beta=0}{{ \beta + m - 1}\choose{m - 1}}.
\end{align*}
\end{proposition}
\begin{proof}
It follows from (\ref{kjnkdsd}) that
\begin{align*}
x(t) \in \operatorname{span} \left( y_{1}, ..., y_{m} \right) 
\end{align*}
at every time $t$, where $\operatorname{span} \left( y_{1}, ..., y_{m} \right)$ denotes the span of the vectors $y_{1}$, ..., $y_m$. We rewrite the function $r_{q, i}\left( x \right)$ in (\ref{r_qi_def}) as
\begin{align*}
r_{q, i}\left( x \right) = r_i(0) + \sum^{q}_{k=1} \sum^k_{\beta=0} \eta\left( k-\beta, \beta \right)  x^{k-\beta}_i x^\beta.
\end{align*}
Both terms $\eta\left( k-\beta, \beta \right)$ and $x^{k-\beta}_i$ are scalars. Thus, we obtain that
\begin{align} \label{dfsdfsdfsdf}
r_{q, i}\left( x \right) = r_i(0) + \sum^q_{\beta=0} \mu_\beta(x) x^\beta
\end{align}
for some scalars $\mu_0(x), \mu_1(x), ..., \mu_q(x) \in \mathbb{R}$. Since the rows of the matrix $F_{q, i}$ are given by (\ref{dfsdfsdfsdf}) for some $x \in \operatorname{span} \left( y_{1}, ..., y_{m} \right)$, it holds that
\begin{align} \label{adfsdfsdfsss}
\operatorname{rank}\left( F_{q,i} \right) \le\sum^q_{\beta=0}\operatorname{dim}\left( \left\{ x^\beta \big| x \in \operatorname{span} \left( y_{1}, ..., y_{m} \right) \right\} \right).
\end{align}
We consider the addends in (\ref{adfsdfsdfsss}) separately. For any vector $x \in \operatorname{span} \left( y_{1}, ..., y_{m} \right)$, it holds that 
\begin{align*}
x^{\beta}_j = \left(\sum^{m}_{p=1}c_{p}\left(y_{p}\right)_j \right)^{\beta}
\end{align*}
for some scalars $c_{1}, ..., c_{m}$. The multinomial theorem yields that
\begin{align*}
x^{\beta}_j = \sum_{p_1+p_2+...+p_{m}=\beta} \frac{\beta!}{p_1! p_2!\cdots p_{m}!}\prod^{m}_{l=1}\left(c_{l} \left(y_{l}\right)_j\right)^{p_l}.
\end{align*}
We define the coefficients
\begin{align*}
\zeta\left( p_1, ..., p_{m} \right) = \frac{\beta!}{p_1! p_2!\cdots p_{m}!}\prod^{m}_{l=1}\left(c_{l}\right)^{p_l},
\end{align*}
which gives that
\begin{align} \label{adfsfsdf}
x^{\beta}_j = \sum_{p_1+p_2+...+p_{m}=\beta} \zeta\left( p_1, ..., p_{m} \right)\prod^{m}_{l=1} \left(y_{l}\right)^{p_l}_j.
\end{align}
By stacking (\ref{adfsfsdf}) for the entries $j=1, ..., N$, we obtain an expression for the vector $x^\beta$ as
\begin{align}\label{kjnkjnss}
x^{\beta} = \sum_{p_1+p_2+...+p_{m}=\beta} \zeta\left( p_1, ..., p_{m} \right)\nu\left( p_1, ..., p_{m} \right) .
\end{align}
Here, we defined the vectors
\begin{align*}
\nu\left( p_1, ..., p_{m} \right) = \left(y_{1}\right)^{p_1} \odot \left(y_{2}\right)^{p_2} \odot ... \odot\left(y_{m}\right)^{p_{m}} ,
\end{align*}
where $\odot$ denotes the Hadamard product, or element-wise product. From (\ref{kjnkjnss}), it follows that the vector~$x^\beta$ is a linear combination of all vectors $\nu\left( p_1, ..., p_{m} \right)$ with $p_1+p_2+...+p_{m}=\beta$, which yields that
\begin{align} \label{afsdfsdfsdfaaa}
\operatorname{dim}\left( \left\{ x^\beta \big| x \in \operatorname{span} \left( y_{1}, ..., y_{m} \right) \right\} \right) = {{ \beta + m - 1}\choose{m - 1}}.
\end{align}
We complete the proof by combining (\ref{adfsdfsdfsss}) and (\ref{afsdfsdfsdfaaa}).
\end{proof}
As an example, consider that the nodal state $x(t)$ of the SIS epidemic process is agitated in only $m=10$ agitation modes $y_{p}$. Then, since $F_i=F_{2, i}$ for the SIS epidemic process, Proposition~\ref{proposition:rank} states that 
\begin{align*}
\operatorname{rank}\left( F_i \right) = \operatorname{rank}\left( F_{2,i} \right) \le \sum^2_{\beta=0}{{ \beta + 10 - 1}\choose{10 - 1}},
\end{align*}
which yields that $\operatorname{rank}\left( F_i \right) \le 66$.

\end{document}